\documentclass[runningheads]{PCIC/llncs}
\usepackage[T1]{fontenc}
\usepackage{amsmath}
\usepackage{amssymb}
\usepackage{amsfonts}
\usepackage{siunitx}
\usepackage{xcolor}
\usepackage{mathtools}
\usepackage{graphicx}
\newtheorem{model}{Model}
\begin{document}
\title{Optimizing Experimental Design for Causal Effect Estimation with Partial Measurements}
\titlerunning{Optimizing Experimental Design}
%
\author{Leopold Mareis\inst{1,2}\orcidID{0000-0002-2624-6522}}
\authorrunning{L. Mareis}
\institute{Fraunhofer Institute for Cognitive Systems IKS, Hansastraße 32, 80686 Munich, Germany \and
TUM School of Computation, Information and Technology, Technical University of Munich, Germany}
\maketitle 
\begin{abstract} 
Instrumental variable regression quantifies causal effects between a possibly confounded treatment variable $ X_2 $ and a response variable $ X_3 $ by leveraging an instrument $ X_1 $. Our work considers the setting where some prior information of the joint distribution of $ X_{123} $ is given, potentially through an initial dataset. However, further samples must be gathered to improve the accuracy of the estimation. We show that under specific parameter configurations in a Gaussian graphical model, taking partial samples from, e.g., $ X_{12} $ can reduce the asymptotic variance of a consistent estimator. This idea is developed by adding a budget constraint over the cost per (partial) sample. The optimization problem is analytically solvable over the real numbers and gives the optimal number of requested partial and complete samples. We provide significance level, power, and sample-size calculations for detecting a non-zero causal effect under optimal budget allocation. 
Our method can considerably reduce the necessary budget and the number of complete samples. Finally, we showcase the advantages and applicability of adaptive causal effect estimation for automotive analytics and pharmaceutical research.

\keywords{Causal Effect Estimation  \and Experimental Design \and Graphical Model.}
\end{abstract}

\section{Introduction}
The problem of forming decisions under uncertainty is central to the success of humans and entities.  If done correctly, the subsequent action of a decision should have a certain effect on the outcome of interest. A common statistical tool to represent and specify the relationships between random variables is a directed acyclic graph (DAG).  It encodes whether direct effects between variables exist and provides a powerful machinery when combined with density estimates.  Estimating causal effects between two variables is essential for understanding the random system, and these finite-sample estimates are uncertain.  In this paper, we study the problem of gathering additional data to increase the accuracy of causal effect estimates.  Inspired by applications, we allow for requesting non-complete, and thus cheaper, observations.  Depending on the domain, cheaper can also mean faster or more ethical.  Our approach leverages asymptotic variances in a Gaussian Graphical Model \cite{yuan2007model} to form decisions on subsequent data requests.  Although this approach can be applied to arbitrary DAGs, we focus on the well-studied instrumental variable (IV) approach \cite{kilbertus2020class,cawley2012medical,staiger1994instrumental}.  If the requirements are satisfied, this setup allows for causal effects estimation without the necessity of randomization.  To make informed decisions, our theory provides information on the interplay between available budget, p-values, and power computations.  This is especially important in safety-critical domains where data sampling is restricted and needs to be approved and assessed.  Most approaches for adaptive data sampling focus on the non-predetermined number of samples and stop the data gathering procedure when pre-specified thresholds are met \cite{murphy2005experimental,pallmann2018adaptive}.  In non-full-information estimation, active learning methods want to uncover the label of samples leading to a high information gain \cite{cohn1996active,saar2004active}.  And finally, \cite{henckel2022graphical} selected statically covariates in valid adjustment sets based on asymptotic variances. 

The paper is structured as follows.  Section \ref{section_more_iv_data} studies the instrumental variable models where more partial instrument-treatment data is helpful.  Section \ref{section_gathering_budget} provides concrete optimized plans for requesting data under a limited budget. This theory is extended in Section \ref{section_power_pval} to detect positive effects via hypothesis testing before two real applications follow in Section \ref{section_applications}.

\section{Setup}
\begin{model} \label{model_instrumental_sem}
    We consider the causal model given by the structural equation model
\begin{equation*} 
    \begin{pmatrix}
    X_1 \\  X_2 \\ X_3 
    \end{pmatrix} =
    \begin{pmatrix}
    0 & 0 & 0 & \\ 
    \lambda_{12} & 0 & 0   \\
    0 & \lambda_{23} & 0  \\
    \end{pmatrix}
    \begin{pmatrix}
    X_1 \\  X_2 \\ X_3 
    \end{pmatrix}
    + 
    \varepsilon
    , \;
    \varepsilon
    \sim 
    N\left( \begin{pmatrix}
    0 \\ 0 \\ 0
    \end{pmatrix}, 
    \begin{pmatrix}
    \omega_{11} & 0 & 0 \\
    0 & \omega_{22} & \omega_{23} \\
    0 & \omega_{23} & \omega_{33} \\
    \end{pmatrix}\right),
\end{equation*}
with $ \lambda_{12} \neq 0, \lambda_{23}, \omega_{23} \in \mathbb{R}$ and $\omega_{11}, \omega_{22}, \omega_{33} > 0$ satisfying $ \omega_{22} \omega_{33} >~\omega_{23}^2 $. 
\end{model}
The variable $ X_1 $ represents the instrument, the variable $ X_2 $ represents the treatment, and the variable $ X_3 $ represents the outcome.  Instead of explicitly modeling a confounding variable, we represent the hidden relationship between $ X_2 $ and $ X_3 $ by $ \omega_{23} $. This is equivalent under the joint linearity assumption, and if $ \omega_{23} \neq 0 $, confounding is present.  Model \ref{model_instrumental_sem} can also be written in matrix notation as $ X = \Lambda^\top X + \varepsilon$ where $ \varepsilon \sim N(0, \Omega)$.  As $ X $ follows a multivariate normal distribution, its variance is given by 
\begin{equation} \label{formula_iv_variance}
      \mathrm{Var}(X) =
    \begin{pmatrix}
    \omega_{11} & \omega_{11}\lambda_{12} & \omega_{11} \lambda_{12}\lambda_{23} \\
    \cdot & \omega_{22} + \omega_{11}\lambda_{12}^2 & \omega_{23} + \lambda_{23} \sigma_{22} \\
    \cdot & \cdot & \omega_{33} + 2\omega_{23}\lambda_{23} + \lambda_{23}^2\sigma_{22}
    \end{pmatrix} \eqqcolon \Sigma.
\end{equation}
Under this setup, the parameter of interest giving the causal effect of $ X_2 $ on $ X_3 $ is $ \lambda_{23} $.  If there exists a connection between the instrument and the treatment, i.e. $ \lambda_{12} \neq 0 $, then \( \hat{\lambda}_{23} = \frac{\hat{\sigma}_{13}}{\hat{\sigma}_{12}} \) is an unbiased estimator \cite{drton2018algebraic}.  This estimator has two properties.  First, it is solely based on entries of the estimated covariance matrix $ \hat{\Sigma} $, and the numerator and denominator can be estimated on distinct or shared datasets.  Secondly, it is unstable when the denominator is close to $ 0 $, so if $ |\omega_{11}\lambda_{12}| $ is small -- in this case we say $ X_1 $ is a \textit{weak instrument}.  Under the weak instrument scenario, our strategy is to measure more data from $ X_{12} $ to reduce the variance of $ \hat{\sigma}_{12} $ and be more confident in its estimate, motivating the next section.

\section{More Instrument Data} \label{section_more_iv_data}

We assume to have $ n_1 \in \mathbb{N} $ initial datapoints of $ X_{123} $ and additionally $ n_2 $ datapoints of $ X_{12} $ which were measured to further investigate the relationship between the instrument and the treatment variable.  With the available information, the estimand $ \hat{\lambda}_{12} $ can be computed based on $ n_1 $ samples for $ \hat{\sigma}_{13}$ and $ n_1 + n_2 $ samples for $ \hat{\sigma}_{12} $.  The following theorem derives its asymptotic variance in terms of $ n_1 $ with Cram\'er's Theorem, and the full proof can be found in Appendix \ref{subsection_appendix_proof_iv_auxillary_variance}.


\begin{theorem}\label{theorem_iv_auxillary12_variance}
Let $ X $ follow Model \ref{model_instrumental_sem}.  Let $ \hat{m}_{jk;1} = \frac{1}{n_{1}} \sum_{i = 1}^{n_{1}} X^{(i)}_{j} X^{(i)}_{k} $ and $ \hat{m}_{jk;2} = \frac{1}{n_{2}} \sum_{i = n_1 + 1}^{n_1 + n_2} X^{(i)}_{j} X^{(i)}_{k} $ be the sample moments of $ X_{j}X_{k} $ on the two datasets.  Furthermore, let the asymptotic relation $ n_2 / n_1 $ be $ \gamma $.  The sample size adjusted estimator for $ \lambda_{23} $ is
\begin{equation*}
\hat{\lambda}_{23} = \frac{ \hat{m}_{13;1}}{\frac{1}{1 + \gamma}  \hat{m}_{12;1} + \frac{\gamma}{1 + \gamma} \hat{m}_{12;2}} 
\end{equation*}
and it attains the asymptotic distribution
\(
    \sqrt{n_1} (\hat{\lambda}_{23} - \lambda_{23}) \overset{d}{\longrightarrow} N(0, v(\gamma))
\)
with the variance 
\begin{align} \label{formula_iv_additional_variance}
    v(\gamma) & = \frac{1}{\sigma_{12}^2} \left( \frac{\sigma_{13}^2 \sigma_{11} \sigma_{22}}{(1 + \gamma) \sigma_{12}^2} 
    - 2 \frac{\sigma_{13} \sigma_{11}\sigma_{23}}{(1 + \gamma) \sigma_{12}} +  \sigma_{11} \sigma_{33} + 2 \sigma_{13}^2 \right) \\ \nonumber
    & = \frac{1}{\omega_{11}\lambda_{12}^2} \left( \frac{\gamma}{(1 + \gamma)}\lambda_{23}^2 \omega_{22} + \frac{(2 + 3 \gamma)}{(1 + \gamma)} \lambda_{23}^2 \omega_{11}\lambda_{12}^2 + \frac{2\gamma}{(1 + \gamma)} \lambda_{23} \omega_{23} +  \omega_{33}  \right).
\end{align}
\end{theorem}

The variance $ v( \gamma ) $ in original parameterization suggests that increasing $ \omega_{22} $ or $ \omega_{33} $ increases the uncertainty of $ \hat{\lambda}_{23} $ in all configurations.  For the remaining original parameters, the variance can decrease or increase, mainly depending on the sign and magnitude of $ \omega_{23} \lambda_{23} $.  Note that $ \omega_{22} \omega_{33} > \omega_{23}^2 $ prevents negative $ v(\gamma)$ if $ \omega_{23} \lambda_{23} $ approaches $ - \infty $.  In the special case where only the primary dataset is available, so $ \gamma = n_2 = 0 $, Theorem \ref{theorem_iv_auxillary12_variance} reduces as follows:

\begin{corollary} \label{corollary_variance_alpha0}
Under the assumptions of Theorem \ref{theorem_iv_auxillary12_variance} with $\gamma = 0$, the estimator $ \hat{\lambda}_{23} $ attains the following asymptotic distribution:
\begin{equation*}
    \sqrt{n} (\hat{\lambda}_{23} - \lambda_{23}) \overset{d}{\longrightarrow} N(0,  2 \lambda_{23}^2 + \frac{\omega_{33}}{\omega_{11} \lambda_{12}^2})
\end{equation*}
\end{corollary}

In the scenario with one available dataset on $ X_{123} $, increasing $  \lambda_{23}^2  $ or $ \omega_{33} $ or decreasing $ \omega_{11} $ or $ \lambda_{12}^2 $, increases the variance.  Intuitively, this means that the weaker the instrument is, the worse the estimand $ \hat{\lambda}_{23} $ behaves.  In the latter case, it can, therefore, be sensible to gather more data on $ X_{12} $:

\begin{corollary}
\label{corollary_iv_minimizing_alpha}
Let $ X $ follow Model \ref{model_instrumental_sem}.  The asymptotic variance function $v(\gamma) $ is decreasing for $ \gamma > 0 $ if and only if
\begin{equation} \label{formula_iv_minimizing_alpha}
    \lambda_{12}^2 < - \frac{\omega_{22} + 2 \omega_{23} / \lambda_{23}}{\omega_{11}}.
\end{equation}
In particular, this can only be the case if $ \omega_{23}\lambda_{23} < 0 $. 
\end{corollary} 

This statement must be understood in the context of the chosen asymptotics, where only the number of $ X_{123} $ samples is in focus. It treats $ n_1 $ samples of $ X_{123} $ equally as $ n_1 $ samples of $ X_{123} $ with additional $ n_2 $ samples of $ X_{12} $.  Consequently, measuring only $ X_{123} $ might not be the variance-reducing strategy as further $ X_{12} $ samples can be requested without penalization.  Particularly, in cases where the sign of $ \omega_{23} / \lambda_{23} $ is negative, so in cases with counteracting effect paths between $ X_2 $ and $ X_3 $, Corollary \ref{corollary_iv_minimizing_alpha} implies that more $ X_{12} $ data is beneficial. Additionally, Equation \ref{formula_iv_minimizing_alpha} implies that small $ \lambda_{12}^2 \omega_{11} $, so a weak instrument scenario, is necessary for additional partial samples.  As an extension and to leverage the reduced cost of measuring only $ X_{12} $, we now add a budget to formulate the problem of informed future data acquisition.

\section{Gathering Further Samples under a Fixed Budget} \label{section_gathering_budget}

Suppose that under Model \ref{model_instrumental_sem}, $ n $ samples of $ X_{123} $ are collected. Further, let there be an additional budget of $ b \in \mathbb{R} $ from which samples of $ X_{123} $ and $ X_{12} $ can be requested for a cost of $ c_1 $, respectively $ c_2 $.  These costs ought to satisfy the relation $ c_2 < c_1 $.  The goal is to spend the available budget in an asymptotically optimal way to reduce the variance of the estimator $ \hat{\lambda}_{12} $. Let $ m_1 $ denote the number of additional $ X_{123} $ samples and $ m_2 $ the number of additional $ X_{12} $ samples.  With the variance function $ v $ from Equation \eqref{formula_iv_additional_variance}, we obtain the optimization problem
\begin{equation} \label{formula_iv_optimization_problem}
    \min_{(m_1, m_2) \in \mathbb{N}^2} \frac{v(\frac{m_2}{n + m_1})}{n + m_1}  \quad \text{, s.t.} \quad  c_1 m_1 + c_2 m_2 \leq b.
\end{equation}
where the objective is derived from the asymptotic variance of $ \hat{\lambda}_{23} - \lambda_{23} $, scaled by the number of complete cases samples $ X_{123} $.  We denote a solution by $ (m_1^\ast, m_2^\ast)$.  

\begin{theorem} \label{theorem_iv_optimization_problem}
Let $ X $ follow Model \ref{model_instrumental_sem}, and let $ b, c_1, c_2 > 0 $ as well as $ n \in \mathbb{N} $ set the Optimization Problem \eqref{formula_iv_optimization_problem}.  The solution to the related continuous optimization problem is given by $ (0, b / c_2) $ if
\begin{equation*}
    \chi_1 \coloneqq (c_{1} - c_{2}) ( \sigma_{11} \sigma_{13}^{2} \sigma_{22} - 2 \, \sigma_{11} \sigma_{12} \sigma_{13} \sigma_{23}) \leq  2 \, \sigma_{12}^{2} \sigma_{13}^{2} c_{2} + \sigma_{11} \sigma_{12}^{2} c_{2} \sigma_{33}\eqqcolon \chi_2.
\end{equation*}
Otherwise, a positive amount of $ X_{12} $ samples are requested and the solution reads
\begin{equation*}
    (\frac{b - \xi c_2 n}{c_1 + \xi c_2}, \frac{ \xi ( b + c_1 n)}{c_1 + \xi c_2}), \text{ where } \xi = \min( -1 + \frac{\sqrt{ \chi_1 \chi_2}}{\chi_2}, \frac{b}{c_2 n}).
\end{equation*}
\end{theorem}
The condition $ \chi_1 \leq \chi_2 $ expressed in original parameterization is
\begin{equation*}
    \frac{c_{1}}{c_2}   (-2 \lambda_{23}\omega_{23} -  \lambda_{23}^2 \omega_{22} - \lambda_{23}^2 \omega_{11}\lambda_{12}^2) \leq 2 \,  \omega_{11} \lambda_{12}^2\lambda_{23}^2 + \omega_{33}
\end{equation*}
and implies three conceptual findings:  Firstly, as expected by Corollary \ref{corollary_variance_alpha0}, additional $ X_{12} $ samples are only beneficial if the causal effect $ \lambda_{23} $ and the confounding strength $ \omega_{23} $ have opposite signs.  Secondly, only the model parameters and the ratio $c_1 / c_2 $ determine whether additional partial samples are requested. And thirdly, weakening the instrument $ \lambda_{12}^2 $ always weakens the condition $ \chi_1 < \chi_2 $ condition in favour of more potential $ X_{12}$ samples.

\begin{example} \label{exapmle_1}
    Let $ X $ follow Model \ref{model_instrumental_sem} with $ \omega_{11} = 0.5 $, $ \omega_{22} = 1$, $ \omega_{33} = 2 $, $ \lambda_{12} = 0.2, \lambda_{23} = 1 $ and $ \omega_{23} = -1 $. In this case, the covariance between the treatment and outcome is $ 0.1 $, underestimating the actual causal effect of $ \lambda_{23} $ by a factor of $ 10 $.  Corollary \ref{corollary_iv_minimizing_alpha} indicates a benefit in an increased number of $ X_{12} $ samples.  Let us further assume that $ n = 500 $ samples were gathered from $ X_{123} $ and that there is an additional budget of $ b = 250 $ from which samples of $ X_{123} $ can be requested for $ c_1 = 3 $ or from $ X_{12} $ for $ c_2 = 1 $.  The floored optimal data request is $ (m_1, m_2) \approx (20, 187) $ giving an approximate variance of $ \approx 17.12\% $ of the estimator $ \hat{\lambda}_{23} $.  In contrast, requesting only samples of $ X_{123} $ gives $ \approx 17.13\% $ variance.  It turns out that $14.8\%$ more budget would be necessary to match the partial measurements estimator's variance.
\end{example}

\section{Sample Size Calculation, Significance Level and Power} \label{section_power_pval}
When gathering valuable, sensitive, or hazardous data, sample size calculations must often be provided in experimental planning.  This ensures enough data is collected for scientific conclusions while minimizing the risk to test subjects and resources.  For this, we consider hypothesis tests of the form $ H_0: \lambda_{23} \leq 0$ against $ H_A: \lambda_{23} > 0 $ to detect wlog a positive effect.  The test family has a rejection criterion of the form $ \hat{\lambda}_{23} > t $.  To represent the uncertainty in the parameter $ \lambda$, let $ v(\gamma; \mu) $ denote the asymptotic variance function obtained by solving the Optimization Problem~\eqref{formula_iv_optimization_problem} where the underlying Model \ref{model_instrumental_sem} satisfies $ \lambda_{23} = \mu  \in \mathbb{R} $.

\begin{lemma}[Significance Level] \label{lemma_significance_level}
    Let $ X $ follow Model \ref{model_instrumental_sem}.  The hypothesis test $ H_0: \lambda_{23} \leq 0$ against $ H_A: \lambda_{23} > 0 $ with the rejection criterion 
    \begin{equation} \label{formula_hypothesis_test_threshold}
        \hat{\lambda}_{23} > \sup_{\mu \leq 0} F^{-1}_{N(\mu, v(\gamma^\ast; \mu) / (n + m_1^\ast))}(1 - \alpha) \eqqcolon \tau
    \end{equation} 
    then attains the asymptotic significance level $ \alpha \in (0,1) $.
\end{lemma}

To enforce $\lambda_{23} = \mu$ in practical calculations, use the estimate $ \hat{\Sigma} $ to identify the underlying parameters\footnote{$\omega_{11} = \sigma_{1 1}, \; \lambda_{12} = \sigma_{12} / \sigma{11}, \; \omega_{22} = - \sigma_{22} - \omega_{11} \lambda_{12}^2, \; \omega_{23} = \sigma_{23} - \sigma_{22} \sigma_{13} / \sigma_{12} , \; \omega_{33} = \sigma_{33} - 2 \omega_{23} \sigma_{13} / \sigma_{12} - \sigma_{13}^2 / \sigma_{12}^2 * \sigma_{22}$.} of $ \Lambda $ and $ \Omega $, and rebuild the model with $\lambda_{23} = \mu$.

\begin{corollary}[Power] \label{corollary_power}
    Let $ X $ follow Model \ref{model_instrumental_sem}.  Let a hypothesis test with threshold $ \tau $ and significance level $ \alpha \in (0,1) $ be given according to Lemma \ref{lemma_significance_level}.  The asymptotic power of that test at effect level $ \mu > 0 $ is
    \begin{equation}
        1 - F_{N(\mu, v(\gamma^\ast; \mu) / ( n + m_1^\ast))}(\tau).
    \end{equation}
\end{corollary}

We recommend analyzing the power of effect levels $ \mu $ around the current estimate $ \hat{\lambda}_{23} $ to cover pessimistic and optimistic scenarios.

\begin{corollary}[$p$-value]
    Let $ X $ follow Model \ref{model_instrumental_sem}.  Let $ \hat{\lambda}_{23} $ be an estimate of $ \lambda_{23} $ on $ n + m_1 $ samples of $ X_{123} $ and $ m_2 $ samples of $ X_{12} $. The p-value of the hypothesis test family from Lemma \ref{lemma_significance_level} is 
\begin{equation} \label{formula_pvalue}
    \max_{\mu \leq 0} \; 1 -  F_{N(\mu, v(m_2 / (n + m_1) ; \mu) / (n + m_1))}(\hat{\lambda}_{23}).
\end{equation}
\end{corollary}

\begin{figure}
    \centering
    \includegraphics[width = 0.35\textwidth]{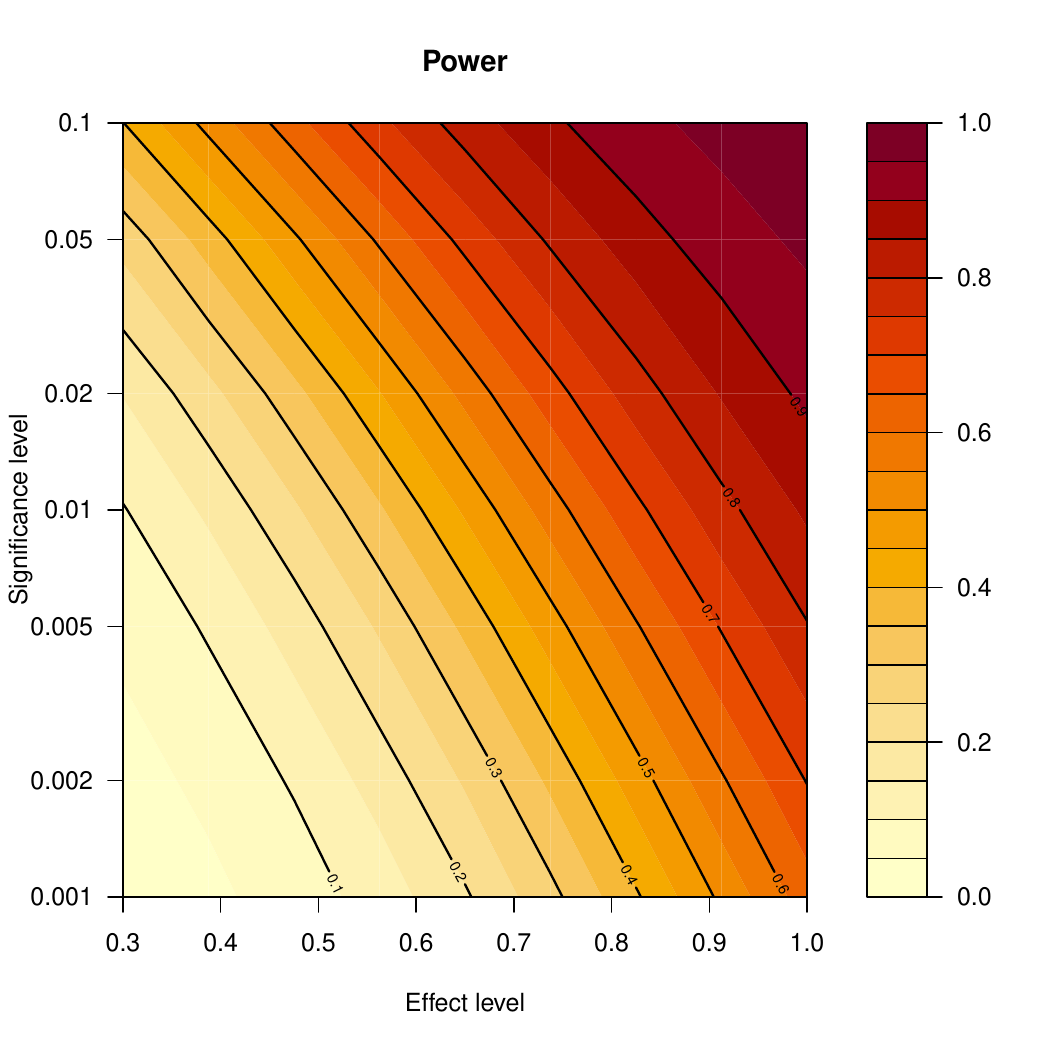} \includegraphics[width = 0.35\textwidth]{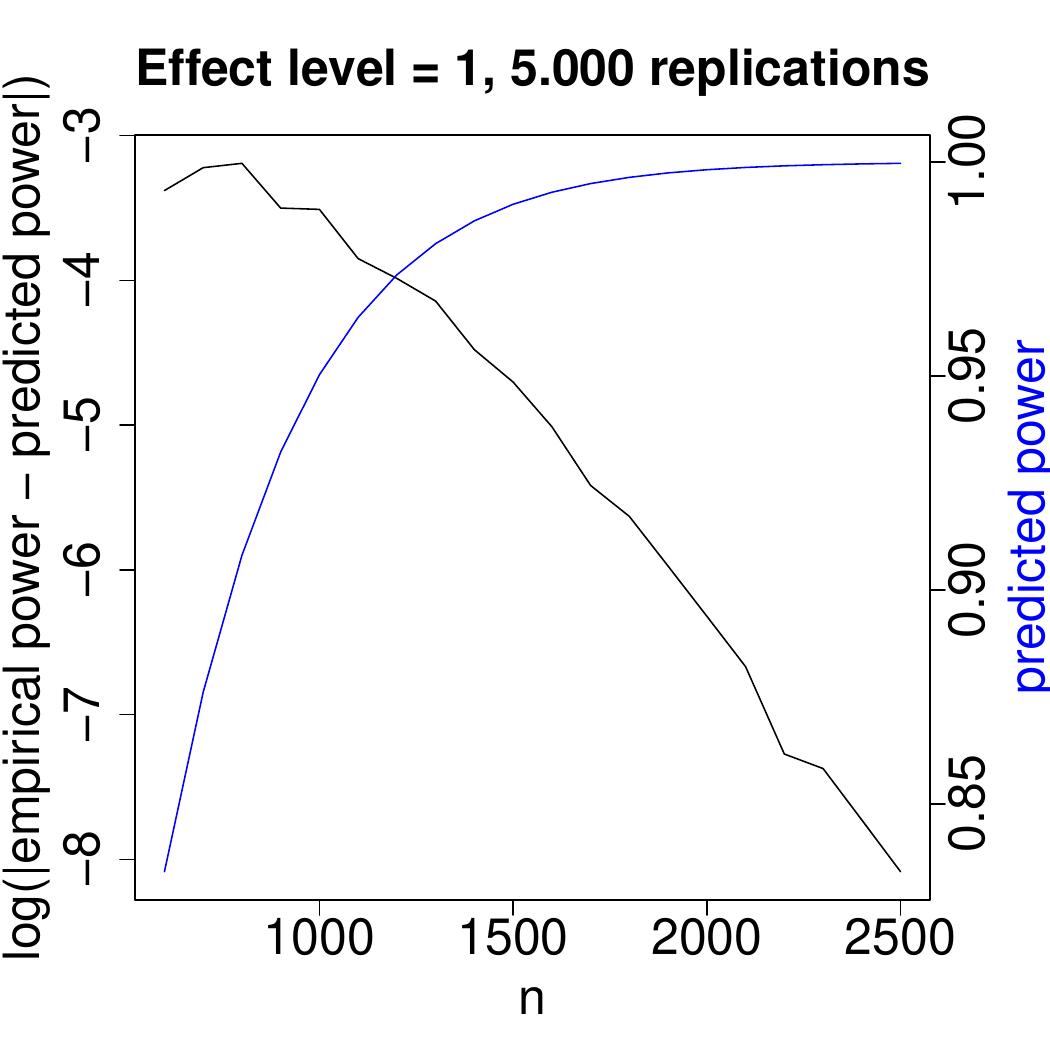}
    \caption{Example 2: Asymptotic power computations and simulations based on Corollary \ref{corollary_power}. Note that both y-axes are logarithmic.}
    \label{figure_power}
\end{figure}

\begin{example}
    Figure \ref{figure_power} shows power computations under the setup of Example \ref{exapmle_1}. The left panel visualizes the power contour lines for combinations of the significance level (y-axis) and the reference effect level of the alternative (x-axis). At the reference effect $ \lambda_{23} > 0.75 $, and significance level $ \alpha = 0.05 $, the hypothesis test's power already lies above $ 0.8 $.  In the right panel, the underlying hypothesis tests' at several sample sizes $ n $ is given in blue according to Corollary \ref{corollary_power}. The black line shows the log distance to the empirical power based on 5000 replications. Under the normality assumptions, the asymptotic approximation of the estimator $\lambda_{23} $ words well.
\end{example}

\section{Real Datasets} \label{section_applications}
\subsubsection{Observed Confounder}
An extended instrumental variable model considers an exogenous additional variable $ X_4 $ influencing $ X_{123} $.
In this scenario, the conditional covariance matrix $ Cov( X_{123} | X_4) $ is equal to the unconditional one from Equation \eqref{formula_iv_variance}.  If additional exogenous information is available, using the conditional covariance matrix through a preprocessing step allows the transfer of previously introduced theory. 

\subsection{Application -- Aggressive Driving}

Aggressive driving is a risk for drivers and positively affects the number of accidents. In this dataset, 10932 drivers were monitored while driving. Environmental variables such as lighting, temperature, humidity, and car-specific variables were recorded. The causal DAG obtained by the PC-Algorithm \cite{kalisch2007estimating} relates the variables and is displayed in Figure \ref{figure_aggressive_dag} in Appendix \ref{appendix_applications}. This DAG suggests an instrumental variable setup with $ (X_1, X_2, X_3) = ($lighting, driving style, speed$)$.  Therefore, the question of interest is whether the driving style influences the car's average speed.  Here, the direct effect is -5.6 km/h average speed per unit increase of aggressive driving style. Hence, actions to relax the driving style should be taken to reduce the average speed, preventing unnecessary accidents \cite{taylor2000effects}.  
When applying our Adaptive Sampling Method~\eqref{formula_iv_optimization_problem}, the competing estimates have the following variances: Optimized variance: 1.012. Baseline variance: 1.274. Theoretical optimized variance: 7.329. Theoretical baseline variance: 7.526. Our proposed optimized method achieves a $21\%$ smaller variance over the baseline procedure.

\subsection{ICU Data}
In the intensive care unit, the water level of patients must be controlled, and Lasix (here Furosemide) is used to drain urine from the patients. Reasons can be malfunctioning kidneys or excess water due to infusions. A preprocessed MIMIC-dataset \cite{johnson2020mimic} on 2000 patients represents an instrumental variable scenario by construction. A patient-wide, caregiver-specific average dosage is used as the instrument; the actual given amount acts as $ X_2 $ and the drained water in ml/min as $ X_3 $. When applying our Adaptive Sampling Method \eqref{formula_iv_optimization_problem}, the competing estimates have the following variances: Optimized variance: $\num{8.879e-09}$. Baseline variance: $\num{2.809e-08}$. Theoretical optimized variance: $\num{1.098e-07}$. Theoretical baseline variance: $\num{1.1038e-07}$. Our proposed optimized method achieves $68\%$ smaller variance over the baseline procedure.

\section{Conclusion}
We have proposed a framework to plan data requests with partial measurements for a given cost function.  It is a conceptually different approach to prior adaptive experimental design methods; however, it keeps the same intention- to reduce or shift the number of necessary total samples. Theorem \ref{theorem_iv_optimization_problem} gives a criterion when partial measurements are helpful and explicitly states the optimal future data request.  When there is no need for partial data, it can indicate aligned confounding and causal effect signs in the underlying model. The central idea of Optimization Problem \ref{formula_iv_optimization_problem}, combining asymptotic variances and budgeting, is not limited to instrumental variable problems and can be derived for more complex causal DAGs. This holds under the premise of identifiable target effects.  Even relaxing the Gaussian Graphical Model assumption and comparing the asymptotic normality of maximum likelihood estimators on partial datasets is feasible.

\begin{credits}
\subsubsection{\ackname} 
This work has been funded by the Federal Ministry of Education and Research (BMBF) as part of AutoDevSafeOps (01IS22087Q) ant the Bavarian Ministry for Economic Affairs, Regional Development and Energy as part of a project to support the thematic development of the Fraunhofer Institute for Cognitive Systems.

\subsubsection{\discintname}
The authors have no competing interests to declare that are relevant to the content of this article. 
\end{credits}
\bibliography{PCIC/bibliography} 
\bibliographystyle{PCIC/splncs04}

\section{Appendix A: Proofs and Derivations}

In this appendix, we collect the proofs from the main document.

\subsubsection{Proof of Theorem \ref{theorem_iv_auxillary12_variance}} \label{subsection_appendix_proof_iv_auxillary_variance}
\begin{proof}
By scaling the asymptotic of $ \hat{m}_{12;2} $ by $ 1/\sqrt{\gamma}$, we obtain
\begin{equation*}
\resizebox{\textwidth}{!}{$
    \sqrt{n_1} \left( 
    \begin{pmatrix}
    \hat{m}_{12;1} \\ \hat{m}_{13;1} \\ \hat{m}_{12;2}
    \end{pmatrix} - 
    \begin{pmatrix}
    \sigma_{12} \\ \sigma_{13} \\ \sigma_{12}
    \end{pmatrix}
    \right) 
    \overset{d}{\longrightarrow}
    N\left( 0, \begin{pmatrix}
    Var(X_1 X_2) & Cov(X_1 X_2, X_1 X_3) & 0 \\
    \cdot & Var(X_1 X_3) & 0 \\
    \cdot & \cdot &  Var(X_1 X_2) / \gamma
    \end{pmatrix}\right).$}
\end{equation*}
As $ X $ is centered, the entries of the asymptotic covariance matrix have the form $ \mathbb{E}[X_i^2 X_j^2] = \sigma_{ii} \sigma_{jj} + 2\sigma_{ij}^2$ and $ \mathbb{E}[X_i^2 X_j X_k] = \sigma_{ii}\sigma_{jk} + 2 \sigma_{ij} \sigma_{ik}$ (see \cite{isserlis1918formula}).  With
\begin{equation*}
    (\hat{m}_{12;1}, \hat{m}_{13;1}, \hat{m}_{12;2}) \mapsto (1 + \gamma)\hat{m}_{13;1} / \left( \hat{m}_{12;1} + \gamma \hat{m}_{12;2}\right) = \hat{\lambda}_{12},
\end{equation*}
the asymptotic distribution of $ \hat{\lambda}_{12} $ must be, by Cram\'er's Theorem (see, e.g., \cite{Ferguson2017}, Thm. 7), Gaussian and its variance can be computed with basic algebra as
\resizebox{\textwidth}{!}{\parbox{\textwidth}{
\begin{align*}
    & (\frac{-\sigma_{13}}{(1 + \gamma) \sigma_{12}^2}, \sigma_{12}^{-1}, \frac{-\gamma \sigma_{13}}{(1 + \gamma) \sigma_{12}^2})
    \begin{pmatrix}
    Var(X_1 X_2) & Cov(X_1 X_2, X_1 X_3) & 0 \\
    \cdot & Var(X_1 X_3) & 0 \\
    \cdot & \cdot &  Var(X_1 X_2) / \gamma
    \end{pmatrix}
    \begin{pmatrix}
    \frac{-\sigma_{13}}{(1 + \gamma) \sigma_{12}^2} \\ \sigma_{12}^{-1} \\ \frac{-\gamma \sigma_{13}}{(1 + \gamma) \sigma_{12}^2}
    \end{pmatrix} \\
    & \quad = \frac{\sigma_{13}^2}{(1 + \gamma)^2 \sigma_{12}^4} \mathbb{E}[(X_1 X_2)^2] 
    - 2 \frac{\sigma_{13}}{(1 + \gamma) \sigma_{12}^3} \mathrm{Cov}(X_1 X_2, X_1 X_3)  \\ & \qquad
    + \frac{1}{\sigma_{12}^2} \mathbb{E}[X_1 X_3]
    + \frac{\gamma \sigma_{13}^2}{(1 + \gamma)^2 \sigma_{12}^4} \mathbb{E}[(X_1 X_2)^2]  \\
    & \quad = \frac{\sigma_{13}^2}{(1 + \gamma)^2 \sigma_{12}^4} \left( \sigma_{11} \sigma_{22} + 2 \sigma_{12}^2 \right)
    - 2 \frac{\sigma_{13}}{(1 + \gamma) \sigma_{12}^3} \left( \sigma_{11}\sigma_{23} +  \sigma_{12}\sigma_{13} \right) \\ & \qquad
    + \frac{1}{\sigma_{12}^2} \left( \sigma_{11} \sigma_{33} + 2 \sigma_{13}^2 \right)
    + \frac{\gamma \sigma_{13}^2}{(1 + \gamma)^2 \sigma_{12}^4} \left( \sigma_{11} \sigma_{22} + 2 \sigma_{12}^2 \right) \\
    & \quad = \frac{1}{\sigma_{12}^2} \left( \frac{\sigma_{13}^2}{(1 + \gamma)  \sigma_{12}^2} \left( \sigma_{11} \sigma_{22} + 2 \sigma_{12}^2 \right)
    - 2 \frac{\sigma_{13}}{ (1 + \gamma)  \sigma_{12}} \left( \sigma_{11}\sigma_{23} +  \sigma_{12}\sigma_{13} \right) 
    + \left( \sigma_{11} \sigma_{33} + 2 \sigma_{13}^2 \right) \right)\\
    & \quad = \frac{1}{\sigma_{12}^2} \left( \frac{\sigma_{13}^2 \sigma_{11} \sigma_{22}}{(1 + \gamma) \sigma_{12}^2} 
    - 2 \frac{\sigma_{13} \sigma_{11}\sigma_{23}}{(1 + \gamma) \sigma_{12}} +  \sigma_{11} \sigma_{33} + 2 \sigma_{13}^2 \right).
\end{align*}}}
This representation of the variance as a function of $ \Sigma $ can now be expressed in the original parameterization of $ \Lambda $ and $ \Omega $. For this, we use Equation \eqref{formula_iv_variance} and find with basic algebra

\noindent
\resizebox{\textwidth}{!}{\parbox{\textwidth}{
\begin{align*}
    & \frac{1}{\omega_{11}^2\lambda_{12}^2} \left( \frac{\omega_{11}^2 \lambda_{12}^2\lambda_{23}^2 \omega_{11} (\omega_{22} + \omega_{11}\lambda_{12}^2)}{(1 + \gamma) \omega_{11}^2\lambda_{12}^2} 
    - 2 \frac{\omega_{11} \lambda_{12}\lambda_{23} \omega_{11} (\omega_{23} + \lambda_{23} (\omega_{22} + \omega_{11}\lambda_{12}^2))}{(1 + \gamma) \omega_{11}\lambda_{12}} \right.  \\ & \qquad \left.
    +  \omega_{11} (\omega_{33} + 2\omega_{23}\lambda_{23} + \lambda_{23}^2(\omega_{22} + \omega_{11}\lambda_{12}^2)) + 2 \omega_{11}^2 \lambda_{12}^2\lambda_{23}^2 \right) \\
    & \quad = \frac{1}{\omega_{11}\lambda_{12}^2} \left( \frac{\lambda_{23}^2 \omega_{22} + \lambda_{23}^2 \omega_{11}\lambda_{12}^2}{(1 + \gamma)} 
    - 2 \frac{\lambda_{23} \omega_{23} + \lambda_{23}^2 \omega_{22} + \lambda_{23}^2\omega_{11}\lambda_{12}^2)}{(1 + \gamma)} \right.  \\ & \qquad \left.
    +  \omega_{33} + 2\omega_{23}\lambda_{23} + \lambda_{23}^2\omega_{22}  + 3 \omega_{11} \lambda_{12}^2\lambda_{23}^2 \right)\\
    & \quad = \frac{1}{\omega_{11}\lambda_{12}^2} \left( \frac{\gamma}{(1 + \gamma)}\lambda_{23}^2 \omega_{22} + \frac{(2 + 3 \gamma)}{(1 + \gamma)} \lambda_{23}^2 \omega_{11}\lambda_{12}^2 + \frac{2\gamma}{(1 + \gamma)} \lambda_{23} \omega_{23} +  \omega_{33}  \right)
    \end{align*}}}
\end{proof}

\subsubsection{Proof of Corollary \ref{corollary_iv_minimizing_alpha}}
\begin{proof}
Taking the derivative gives
\begin{align*}
    & v^\prime(\gamma) =\frac{1}{(1 + \gamma )^2} \frac{1}{\omega_{11}\lambda_{12}^2} \left( \lambda_{23}^2 \omega_{22} + \lambda_{23}^2 \omega_{11}\lambda_{12}^2 + 2\lambda_{23} \omega_{23}  \right) < 0 \\
    & \iff \omega_{22} +  \omega_{11}\lambda_{12}^2 + 2 \omega_{23} / \lambda_{23} < 0 \\
    & \iff \lambda_{12}^2 < - \frac{\omega_{22} + 2 \omega_{23} / \lambda_{23}}{\omega_{11}}.
\end{align*}
\end{proof}

\subsubsection{Proof of Theorem \ref{theorem_iv_optimization_problem}}
We focus on the Optimization Problem \eqref{formula_iv_optimization_problem} with the continuous constraint $ (m_1, m_2) \in [0, \infty)^2 $, so 
\begin{equation} 
    \min_{(m_1, m_2) \in [0, \infty)^2} \frac{v(\frac{m_2}{n + m_1})}{n + m_1}  \quad \text{, s.t.} \quad  c_1 m_1 + c_2 m_2 \leq b.
\end{equation}
\begin{proof}
For a fixed ratio $  m_2 / (n + m_1) $, more samples will always decrease the objective as $ v( m_2 / (n + m_1)) $ is constant and the denominator increases. Therefore, only the boundary has to be checked, and we can require $ c_1 m_1 + c_2 m_2 = b $.  With the relation $ m_2 = (b - c_1 m_1) / c_2$, the problem 
\begin{equation} \label{formula_iv_optimization_problem_versionII}
    \min_{m \in [0, b / c_1]} \frac{v(\frac{1}{c_2}\left( \frac{b - c_1 m}{n + m} \right))}{n + m} \eqqcolon h(c_1, c_2, b, n, m),
\end{equation}
can be solved equivalently.  The solution to the original problem is then given by $ (m^\ast, (b - c_1 m^\ast) / c_2) $. We can equate the derivative $ \partial  h(c_1, c_2, b, n, m) / \partial m $ to $ 0 $ and find
\begin{align*}
    &\frac{v^\prime(\frac{1}{c_2}\left( \frac{b - c_1 m}{n + m} \right)) \frac{1}{c_2} \frac{-c_1 (n + m) - (b - c_1 m)}{(n + m)^2} (n + m) - v(\frac{1}{c_2}\left( \frac{b - c_1 m}{n + m} \right))}{(n + m)^2} \overset{!}{=} 0 \\
    &\Leftrightarrow v^\prime(\xi) \left( \frac{c_1}{c_2} + \xi \right) + v(\xi) = 0 \quad \text{, where } \quad \xi = \frac{1}{c_2}\left( \frac{b - c_1 m}{n + m} \right).
\end{align*}

Let's define $ \chi_1 \coloneqq (c_{1} - c_{2}) ( \sigma_{11} \sigma_{13}^{2} \sigma_{22} - 2 \, \sigma_{11} \sigma_{12} \sigma_{13} \sigma_{23}) $ and $ \chi_2 \coloneqq 2 \, \sigma_{12}^{2} \sigma_{13}^{2} c_{2} + \sigma_{11} \sigma_{12}^{2} c_{2} \sigma_{33}$. Then, the necessary equation simplifies to 

\begin{equation*}
    v^\prime(\xi) \left( \frac{c_1}{c_2} + \xi \right) + v(\xi) = \frac{\chi_2 \xi^{2}}{\sigma_{12}^{4} c_{2} (1 + \xi)^{2}} + \frac{ 2 \, \chi_2 \xi}{\sigma_{12}^{4} c_{2} (1 + \xi)^{2}} + \frac{\chi_2 - \chi_1}{\sigma_{12}^{4} c_{2} (1 + \xi)^{2}}
\end{equation*} with solutions $ -1 \pm \sqrt{ \chi_1 \chi_2} / \chi_2 $. So $ \xi > 0 $, which is necessary for $m^\ast > 0 $, is true if and only if $ [(c_{1} - c_{2}) / c_2] ( \sigma_{13}^{2} \sigma_{22} - 2 \,\sigma_{12} \sigma_{13} \sigma_{23}) > 2 \, \sigma_{12}^{2} \sigma_{13}^{2}  / \sigma_{11} +  \sigma_{12}^{2}  \sigma_{33} $, or equivalently $ (c_1 / c_2) (-2 \lambda_{23}\omega_{23} -  \lambda_{23}^2 \omega_{22} - \lambda_{23}^2 \omega_{11}\lambda_{12}^2)) > 2 \,  \omega_{11} \lambda_{12}^2\lambda_{23}^2 + \omega_{33}$.
\end{proof}

\subsubsection{Proof of Lemma \ref{lemma_significance_level}}
\begin{proof}
    Let $ \mu \leq 0 $ be an arbitrary parameter in the null hypothesis for the model parameter $ \lambda_{23} $.  Following the recommendation of the Optimization Problem~\eqref{formula_iv_optimization_problem}, we have that $ \hat{\lambda}_{23} $ is asymptotically $ N\left( \mu, v(\gamma; \mu) / ( n + m_1)\right) $ distributed.  Thus $ \lim_{n \to \infty} P_{\lambda_{23} = \mu}(\hat{\lambda}_{23} > F^{-1}_{N(\mu, v(\gamma; \mu) / (n + m_1))}(1 - \alpha)) = \alpha $ and the asymptotic significance level is given by
    \begin{align*}
        \lim_{n \to \infty}P_{H_0}(\mathrm{reject } \; H_0) 
        &= \lim_{n \to \infty} \sup_{\mu \leq 0} P_{\lambda_{23} = \mu}\left(\hat{\lambda}_{23} > \tau \right) \\ 
        &\leq \sup_{\mu \leq 0} \lim_{n \to \infty} P_{\lambda_{23} = \mu}\left(\hat{\lambda}_{23} > F^{-1}_{N(\mu, v(\gamma; \mu) / (n + m_1))}(1 - \alpha)\right) =  \alpha.
    \end{align*}
\end{proof}

\section{Appendix B: Applications -- Further Details} \label{appendix_applications}
\subsubsection{Aggressive Driving}
\begin{figure}[h] \label{figure_aggressive_dag}
    \centering
    \includegraphics[width = 10cm]{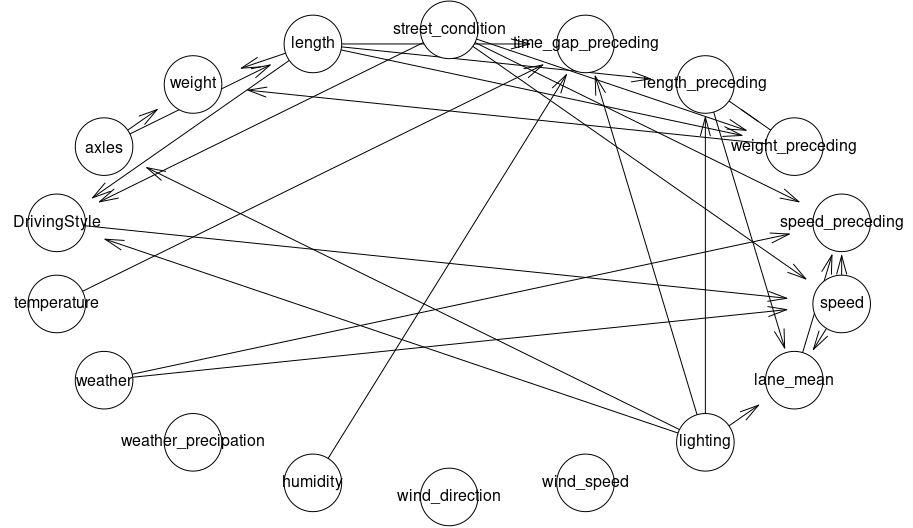}
    \caption{This DAG on the Aggressive Driving was estimated by the PC-Algorithm.}
\end{figure}
All measured variables from the publicly available aggressive driving dataset\footnote{\url{https://www.kaggle.com/datasets/veeralakrishna/aggressive-driving-data}} are related in the learned DAG in Figure \ref{figure_aggressive_dag}.  The measurement 'driving style' has three levels: \textit{aggressive}, \textit{normal} and \textit{relaxed} driving. For the adaptive data collection, we assumed $ 500 $ initial data samples on $ X_{123} $, an additional budget of $ 300 $, and a cost per sample of $ 1 $ for $ X_{123} $ and $ 0.3 $ for $ X_{12} $.   The optimized allocation is $ m = (211, 293) $.  The variances were estimated using subsampling on the dataset. The estimated covariance matrix over the full dataset has the following entries:
\begin{equation*}
\begin{matrix}
    \text{Lighting} \\ \text{Driving Style} \\ \text{Speed}
\end{matrix}
    \begin{pmatrix}
  \num{7.436e-01} &  \num{6.806e-02} &  \num{2.225e-01}     \\
6.806 \times 10^{-2} &  \num{5.112e-01} & \num{-8.365e+00}    \\
 \num{2.225e-01} & \num{-8.365e+00} & \num{3.531e+01}  
\end{pmatrix}
\end{equation*} 

\subsubsection{ICU Medication}
20442 patients were treated in the MIMIC-IV ICU database with Lasix, where foley output (urine) was measured during the subsequent 6 hours.  To eliminate dependencies between samples, we selected every patient’s first ICU stay. As the response variable $X_3$, we selected the \textit{urine}/\textit{minutes passed} ratio of the last measurement within these 6 hours. We selected 2000 patients at random and used their caregiver ID to generate an instrumental caregiver attribute over the remaining 18442 patients. This was the average prescribed amount of the caregiver in the secondary dataset.  The estimated covariance matrix over the full dataset has the following entries:
\begin{equation*}
\begin{matrix}
    \text{Caregiver Rate} \\ \text{Medication} \\ \text{Foley Rate}
\end{matrix}
    \begin{pmatrix}
      \num{1.345e+04} & \num{7.921e+03} & \num{0.2424} \\
         \num{7.921e+03} & \num{1.125e+05} &  \num{1.038} \\
    \num{2.424e-01} & \num{1.037} &  \num{0.1082} 
\end{pmatrix}
\end{equation*} 
For the adaptive data collection, we assumed 200 initial data samples on $X_{123}$, an additional budget of 100, and a cost per sample of 1 for $X_{123}$ and 1/7 for $X_{12}$. The optimized allocation is $ m = (88, 77) $ with an effect size of $ \num{3.064e-5}$.  The variances were estimated using subsampling on the remaining dataset.

\subsubsection{Assembly Line}
In the previous two scenarios, our method advised gathering more partial data. A real-life covariance matrix where this is not the case is       
\begin{equation*}
    \begin{pmatrix}
  \num{7.278e-10} & \num{-1.451e-07} & \num{-9.014e-08} \\
\num{-1.450e-07} & \num{ 2.429e+00} & \num{-4.563e-01} \\
 \num{-9.014e-08} & \num{-4.564e-01} &\num{ 7.330e-01}
\end{pmatrix}.
\end{equation*}
The underlying dataset is presented in \cite{gobler2024textttcausalassembly} and combines 15581 data points and expert knowledge from an assembly line. Thus, we can ensure that the instrumental variable assumptions are satisfied.

\end{document}